\documentclass[a4]{amsart}
\usepackage{graphicx}
\usepackage{amssymb}
\graphicspath{{.}{fig/}}
\usepackage{enumerate}
\usepackage{eepic}

\newtheorem{theorem}{Theorem}

\newtheorem{definition}{Definition}

\usepackage[linesnumbered,boxed,ruled,algonl]{algorithm2e}
\SetKwInput{KwInit}{Initialization}

\begin{document}
%\title{An algorithm to find cliques in graphs}
\title{Is it possible to find the maximum clique in general graphs?}
\author[J.I. Alvarez-Hamelin]{Jos\'e Ignacio Alvarez-Hamelin}
\address{INTECIN (UBA-CONICET), Facultad de Ingenier\'{\i}a, Paseo Col\'on 850, C1063ACV Buenos Aires -- Argentina}
\email{ignacio.alvarez-hamelin@cnet.fi.uba.ar}
\thanks{{\em Acknowledgements}: Juan Ignacio Giribet, Pablo Jacovkis,
Mario Valencia-Pabon, Ariel Futoransky}
%\date{September 16, 2011} %v1
%\date{September 26, 2011} %v2
%\date{Febraury 1, 2012} %v4
\date{\today}

\maketitle{}

\begin{abstract}
Finding the maximum clique is a known NP-Complete problem and it is also
hard to approximate. This work proposes two efficient algorithms to obtain
it. Nevertheless, the first one is able to fins the maximum for some
special cases, while the second one has its execution time bounded by the
number of cliques that each vertex belongs to. 
\end{abstract}
\section{Introduction}

Finding cliques in graphs is a well known problem, mainly the maximum
clique was found to be a NP-Complete problem~\cite{Kar72}.  Indeed, from
any vertex, to discover the maximum clique to which it belongs, we
should take all combinations of $k$ neighbors and verify whenever they
are mutually adjacent, which yields an exponential time as a function of
the vertex degree. Moreover, \cite{Johnson:1973} shows that there exists
no sublinear approximation algorithm. This problem was largely treated
and an extensive survey can be found in~\cite{bomze1999}.  In this work,
we present two algorithms to find maximum cliques. We decompose the
graph vertex by vertex until no vertices remain; then, we re-build the
graph restoring each of the vertices, one by one, in an inverted order
computing maximal cliques at each step.  Next section is devoted to
present the algorithms and the theorems showing their correctness. A section
showing the algorithms applied to real graphs is presented to illustrate
how they works. The paper is concluded with a discussion about the
complexity of the problem.

\section{Algorithms}

Let $G=(V,E)$ be a simple undirected graph with $n=|V|$ vertices and $m=|E|
\leq |V|\times|V|$ edges. The neighborhood of a vertex is the set composed by
vertices directly connected to it, i.e., $w\in N(v)$ such as $\{v,w\}\in E$.
Then, the degree of vertex $v$ is denoted as $d(v)=|N(v)|$. Let us call
$H=(V,E,A)$ the annotated graph $G$, where $|A|=|V|$ such that if $v\in V$ then
$a_v$ is a list of attributes of vertex $v$.
%In this work, each of these
%attributes is a list of sets called $a_v^i$, and $i\in N(v)$. Abusing the
%notation, we use $a_v=\{a_v^x, \forall x \in N(v)\}$, and $a_v=\emptyset$ when
%$N(v)=\emptyset$.  

Attributes $a_v$ are list of sets, each of one a clique, and they are
computed by the proposed algorithms.
For Algorithm~\ref{cliques2}, we denote an element of the list as a set
$L \in a_v^x$, its initialization with set $L$ as
$a_v^x\leftarrow\{L\}$, and the append of set $L$ as
$a_v\leftarrow \{a_v, L \}$ and the elimination of a set L as
$a_v\leftarrow\{a_v\setminus L\}$. The elements of this list can also
interpreted as $a_v^i$ for $i\in N(v)$ (see Algorithm~\ref{cliques}).  
%We denotes $A\in a_v^x$ when $A$ is one of
%the sets in the list $a_v^x$. 
A list of maximal cliques is stated by Definition~\ref{def:max}.

\begin{definition}\label{def:max}
Given a vertex $v$ and the list $l_v$ where its elements are sets $S_i$ having $v$, i.e.,
$S_i \supset v$. These sets $S_i \in l_v$ are maximal if for all $i\neq j$ they
verify the following properties: 
\begin{enumerate}
  \item $v\subseteq S_i\cap S_j $
  \item $S_i\nsubseteq S_j$
  \item $S_i\nsupseteq S_j$
  \end{enumerate}
\end{definition}

Therefore, lists $a_v$ have the properties presented in the following
Definition~\ref{def2}.

\begin{definition}\label{def2}
Let a vertex  $v\in V$  the $a_v$ be a list composed by sets of vertices
$A_i$, such that the following properties hold:
  \begin{enumerate}
  \item each set $A_i\in a_v$ denotes a {\em maximal} clique having $v$ (see Definition~\ref{def:max});~\label{def2:p1}
  \item the maximum clique $K_{\max}(G)$ of the graph $G$ is
found as\label{def2:p2}
   \begin{equation}
   K_{\max}(G)= A \;\;\mathrm{if}\; A=\max\big( |A_i|: A_i\in a_v i\in [1,|a_v|] \big) \enspace,\label{eq_MK}
   \end{equation}
  \end{enumerate}
\noindent where $|a_v|$ refers to the length of $a_v$ list.
\end{definition}
Notice, firstly, each set $A_i$ is maximal in the sense that there not exist
other set $B\supseteq A_i$; secondly  sets $A_i\in a_v$ are maximal cliques at
the time that they are computed (see Algorithm~\ref{cliques2}), but at any later
time can exist other clique maximal containing the $A_i$. As we demonstrate,
computing Equation~\ref{eq_MK} when Algorithm~\ref{cliques2} finished leads to
obtain the maximum clique. We analyze it cost later.

We introduce first an algorithm of low complexity that could find the maximum clique in certain cases, this is the Algorithm~\ref{cliques}.

\begin{algorithm}[t]
 \caption{Function $H\leftarrow${\tt find\_cliques}$(G)$}\label{cliques}
 \SetLine
  \KwIn{a graph $G=(V,E)$}
  \KwOut{a graph $H=(V,E,A)$ }
  \Begin{
     find a vertex $v$ such that $d(v)$ is maximum~\label{algo:deg}\;
     set $M=N(v)$\;
     $G' \leftarrow (V',E'):  V'=V\setminus v, \; E'=E\setminus\{v\times N(v)\}$~\label{algo:minusv}\;
     $a_v^x\leftarrow\emptyset$ for all $x\in M$, or $a_v\leftarrow\emptyset$ if $M=\emptyset$~\label{algo:ini}\;
     set $H'\leftarrow\{\emptyset,\emptyset,\emptyset$\} \; 
     \If {$G'\neq \emptyset$~\label{algo:endif}}{
         set $H' \leftarrow {\tt find\_cliques}(G')$~\label{algo:recur}\;
         \For{each $x\in M$~\label{algo:forneig} }{
	     \For{each set $a_x^i$\label{algo:neigL}, or $a_x=\emptyset$}{
		     {\bf if} ($a_x\neq\emptyset$), {\bf then} $L\leftarrow \left(N(v) \cap a_x^i\right) \cup v \cup x$, {\bf else} $L\leftarrow v \cup x$~\label{algo:seta} \;
                     \If {$(|L| \geq |a_v^x|$~\label{algo:des} or $a_x=\emptyset)$}{
		             $a_x^v\leftarrow L$~\label{algo:neig}\; 
			     $a_v^x\leftarrow L$~\label{algo:local}\;
                     }
	     }
         }
     }
     set $H\leftarrow H' \cup (v,v\times M,a_v)$~\label{algo:valueH}\;
     \Return{H}~\label{algo:return}\;
  }
\end{algorithm}

\begin{theorem}~\label{main}
Algorithm~\ref{cliques} ends if $G=(V,E)$ has finite size, 
computing the attributes $a_v^x$ for all vertices $v\in V$ and
their corresponding neighbors, according to the
property~{\em(\ref{def2:p1})} in Definition~\ref{def2}.
\end{theorem}
\begin{proof}
Let us start with the winding phase of the recursion, that is,
steps~\ref{algo:deg},~\ref{algo:minusv} and~\ref{algo:recur} are
executed until we reach an empty graph $G'$ (see~\ref{algo:endif}).
At this step, the end of recursion (step~\ref{algo:endif}) is found because each call
to {\tt find\_cliques$(G')$} function is done with a reduced set of vertices
$|V'|=|V-1|$ (and its induced graph), and $G$ has finite size; i.e., the
recursion is done $n$ times. 

From there, we analyze the unwinding phase. Let's start when
steps~\ref{algo:valueH} and~\ref{algo:return} are executed for the first
time: the return of the function will carry a graph with just the vertex
$v$, no edges and $a_v=\emptyset$ (step~\ref{algo:ini}), that is
$H=(\{v\},\emptyset,a_v)$.  Then, the following instance(s) can add
vertices of degree zero until one instance begins to add the first edges
(edge), getting a star with leaves (one leaf), because the degree is an
increasing function (the winding phase was carried out taking the
maximum degree at step~\ref{algo:deg}, so the unwinding one reconnects
vertices with the same degree or greater one).  At this instance,
steps~\ref{algo:forneig} and~\ref{algo:neigL} are executed and
step~\ref{algo:seta} yields $L=\{v,x\}$ because $a_x$ is empty (the
vertex $x$ has no registered neighbors until now).  The following
conditional sentence is true (step~\ref{algo:ini} assures
$a_v^x=\emptyset$) setting each neighbor as a clique, on both sides,
the neighbor vertex $a_x^v=\{v,x\}$ and the local one $a_v^x=\{v,x\}$
(see~\ref{algo:des}, \ref{algo:neig} and~\ref{algo:local}, we consider
objects $a_w$ as mutable).

From this point of the algorithm execution any one of the next vertices
could build a $K_s/s\in\mathbb{N}$ (e.g., $s=3$) because a new vertex
joining former vertices constituting a $K_{s-1}$ could appear. It is
worth remarking that it is not possible to build a $K_{s+1}$ at this
stage (e.g., $s+1=4$). The reason why it is not possible is that there are
only $K_{s-1}$ and a new vertex just adds edges between this new vertex
and the present vertices, although this new vertex will never add an
edge between the present vertices. In this way, the size of new cliques
is an increasing function (either the maximum clique remains at same
size or it is increased by one vertex).

Now, we will show the $a_v^x$ is always a clique.  We have also shown that the
first elements in $a_v^x$ constitute a clique of two vertices: $v$ and its
neighbor $x$.  Considering, at any instance in the unwinding phase of the
Algorithm~\ref{cliques}, a vertex $v$ has a clique stored for each one of its
neighbors $x$ in $a_v^x$. Let's consider, without loss of generality, a new
neighbor of $v$, called $w$, having as neighbors $B\subseteq a_v^x=K_t$, that
is $N(w)\supseteq B\cup v$. When step~\ref{algo:des} is executed, either
$a_w^v=\emptyset$ or $a_w^v=\{w,v\}$ (because it is possible that $a_v^x\cap
N(w)$ is empty in~\ref{algo:seta}), and then $a_w^v\leftarrow B \cup w$, which
is also a clique because $B\subseteq a_v^x$ is a clique and $w$ is a neighbor
of all vertices in $B$ by hypothesis. 

To conclude the proof is enough to determine if the
property~{\em(\ref{def2:p1})} in Definition~\ref{def2} is obtained by
Algorithm~\ref{cliques}.  The initial case was already shown in the second
paragraph of this proof. Then, considering a case where a maximal clique of
vertices $w$ and $y$ is $K_s$, and there exists another clique $K_t$ such that
$w,y \in K_t$ and $t\leq s$. As seen in a previous paragraph, vertices can only
build cliques that increase the previous one by just one vertex.  Let's
consider that the next vertex $z$ is connected with all vertices in $K_t$ and
$z$ is at most connected with $t-1$ vertices in $K_s$. The minimum difference
that is needed to distinguish between two cliques is one vertex. Taking into
account, without loss of generality, that the maximum clique for $z$ is
$K_s\cup z$, steps~\ref{algo:seta} and~\ref{algo:des} will select the clique
$K_s\cup z$ because at least one neighbor of $z$, let's call it $u$, has
$a_u^i=K_s/i\in N(u)$.    At this point, the values $a_z^w$, $a_z^y$, $a_w^z$,
and $a_y^z$  will be updated because their size is greater or equal to the size
of a clique found before (see~\ref{algo:des}, \ref{algo:neig}
and~\ref{algo:local}).  Thus, if $t+1<s$ then $K_s$ remains a maximal clique
for $w$ and $y$; or else $t+1\geq s$ and the set of $t$ vertices in $K_t$ plus
$z$ is a maximal clique for $w$ and $y$.  It is worth remarking that it is
possible that a vertex $v$ has several cliques with a neighbor $x$, and the
condition in~\ref{algo:des} assures that the maximum clique, among the known
cliques, is taken.  Notice that $a_w^y$ and $a_y^w$ have the other clique $K_s$
stored, but still $K_{t+1}=a_w^x \max(|a_w^x|, \forall x\in N(v))$ is the
maximum clique, among the known cliques, for $w$ because $t+1\geq s$ and the
previous maximum clique was $K_s$; the same occurs to $y$. 
\end{proof}

Before to analyze its complexity, we remark that Algorithm~\ref{cliques}
not guarantee that the operation  $A=\max\big( |A_i|: A_i\in a_v^x, v\in
V, x\in N(v) \big)$ gives the maximum clique $K_{\max}(G)$.  In fact,
Figure~\ref{fig:ctex1} presents a counterexample. Imagine that the
unwinding phase takes vertices in the order 1, 2, 3, 4, 5, 6, 10, 11,
12, 13, 14.  Regarding that list $a_{12}$  has not the clique
$\{10,11,12\}$ because $a_{11}^{10}=\{5,6,10,11\}$,
$a_{12}^{11}=\{3,4,11,12\}$ and $a_{12}^{10}=\{1,2,10,12\}$; therefore,
when vertex 13 is added it never find the clique $\{10,11,12,13\}$, so
the maximum clique is not found.  Nevertheless, the maximum clique is
often found when very few vertices have degree close to $d_{\max}$ and
the graph is sparse, as  it is the case for the so called scale free
networks~\footnote{Graphs having a heavy tailed degree distribution
which can be bound by a power law.}. 

\begin{figure}
\includegraphics[width=6cm]{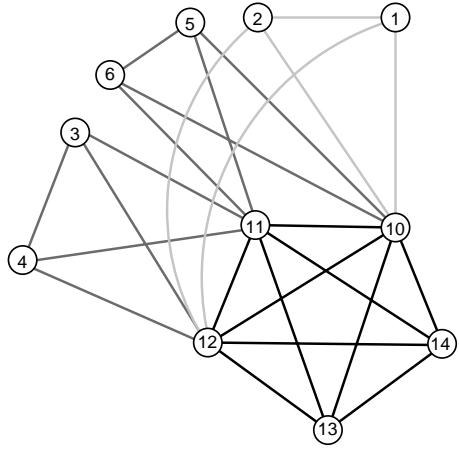}
\caption{Counterexample for Algorithm~\ref{cliques}, which gives one of size 4 instead of 5 when vertices are treated an increasing order.\label{fig:ctex1}}
\end{figure}

\paragraph{\bf Time complexity of Algorithm~\ref{cliques}}
Step~\ref{algo:seta} of Algorithm~\ref{cliques} is an intersection of
two sets of size $d_{\max}$ (maximum degree), if both are
ordered\footnote{This can be done at the beginning for all vertices in
graph $G$, taking $O(n\cdot d_{\max} \cdot \log(d_{\max}))$.} taking
$O(d_{\max})$.  Next, we consider loop~\ref{algo:neigL}, taking
$d_{\max}$ times, what determines $O(d_{\max}^2)$. 
Loop~\ref{algo:forneig} on vertex $v$ neighbors~\ref{algo:forneig} takes
an extra $d_{\max}$, giving $O(d_{\max}^3)$.  Step~\ref{algo:deg} has a
complexity of $O(n)$ if vertices are not ordered, but there is an
additive cost expected smaller than $O(d_{\max}^3)$.  Finally, the
recursion is done for each vertex of the graph, producing a total time
complexity of $O(n\cdot d_{\max}^3)$.  In the case when neighbors are
not ordered, we get $O(n\cdot d_{\max}^4)$. Considering connected graph,
we can express the $n$ recursions and the visit to all neighbors in
loop~\ref{algo:forneig} as visiting all the edges, yielding a time
complexity of $O(m\cdot d_{\max}^2)$.

For graphs in general, where $d_{\max}$ could be bound by $n$, time
complexity is $O(n^4)$. However, for graphs having a heavy tailed
degree distribution which can be bound by a power law\footnote{Most of
the real problems in Complex Systems field has $2\leq\beta\leq3$.}:
$P(d)\propto d^{-\beta}$, it is possible to find a lower bound. Indeed,
these graphs have $n^{\frac{1}{\beta}}$ as a bound of $d_{\max}$,
therefore the complexity yields $O(n^{1+\frac{3}{\beta}})$ for
$\beta\leq 3$; and for $\beta>3$ either the search in~\ref{algo:deg} or
the elimination in \ref{algo:minusv} dominates, reaching the bound
$O(n^2)$ or $O(n\cdot d_{\max}\cdot\log d_{\max})$ respectively. 

\paragraph{\bf Storage complexity} It can be computed as the space to
storage graph $G$, which is $O(n\cdot d_{\max})$, and the space
occupied by all the $a_v^x$ sets. This last quantity can be computed as
the length of each set $a_v^x$, which is bound by $d_{max}$ and the
number of them per vertex, which is also bound by $d_{max}$, yielding 
$O(d_{\max}^2)$ per vertex. Thus, the total storage complexity is
$O(n\cdot d_{\max}^2)$.

\vspace{5mm}
Now, we present another Algorithm~\ref{cliques2} capable to find the maximum
clique in any graph.

\begin{algorithm}[t]
 \caption{Function $H\leftarrow${\tt find\_cliques}$(G)$}\label{cliques2}
 \SetLine
  \KwIn{a graph $G=(V,E)$}
  \KwOut{a graph $H=(V,E,A)$ }
  \Begin{
     find a vertex $v\in V$~\label{algo2:deg}\;
     set $M=N(v)$\;
     $G' \leftarrow (V',E'):  V'=V\setminus v, \; E'=E\setminus\{v\times N(v)\}$~\label{algo2:minusv}\;
     $a_v\leftarrow\emptyset$~\label{algo2:ini}\;
     set $H'\leftarrow\{\emptyset,\emptyset,\emptyset$\} \; 
     \If {$G'\neq \emptyset$~\label{algo2:endif}}{
         set $H' \leftarrow {\tt find\_cliques}(G')$~\label{algo2:recur}\;
         \For{each $x\in M$~\label{algo2:forneig} }{
	     \For{each set $A\in a_x$\label{algo2:neigL}, or $a_x=\emptyset$}{
		     set $new\leftarrow FALSE$, and $old\leftarrow FALSE$~\label{algo2:iniF}\;
		     \lIf{$a_x\neq\emptyset$}{$L\leftarrow \left(N(v) \cap A\right) \cup v \cup x$ }\lElse{$L\leftarrow v \cup x$~\label{algo2:seta}} \;
		     \For{each set $B\in a_v$, or $a_v=\emptyset$\label{algo2:localL}}{
			     \eIf{$|B|<|L|$~\label{algo2:local-s}}{
				     \If{$B\neq \emptyset$ and $B\subset L$~\label{algo2:local-in-L}}{
                                             $a_v\leftarrow\{a_v\setminus B\}$~\label{algo2:elimB}\;
					     %\lIf{$B\neq \emptyset$}{$a_v\leftarrow\{a_v\setminus B\}$}~\label{algo2:elimB}\;
				     }
				     set $new=TRUE$~\label{algo2:BmL}\;
			     }{       
				     \eIf{$B\nsupseteq L$\label{algo2:BnotinL}}{
					     set $new\leftarrow TRUE$~\label{algo2:BnotinL-T}\; 
				     }{
					     set $old\leftarrow TRUE$~\label{algo2:oldT}\;
				      }
                             }
                       }\label{algo2:localL-end}
		     \lIf {$new == TRUE  \;\&\;  old == FALSE$}{ $a_v\leftarrow \{a_v,L\}$\label{algo2:selectL}}
	     }
         }
     }
     set $H\leftarrow H' \cup (v,v\times M,a_v)$~\label{algo2:valueH}\;
     \Return{H}~\label{algo2:return}\;
  }
\end{algorithm}

The following theorem proof the correctness of 
Algorithm~\ref{cliques2}, verifying how this algorithm accords with 
Definition~\ref{def2}.

\begin{theorem}~\label{main2}
Algorithm~\ref{cliques} ends if $G=(V,E)$ has finite size, 
computing correctly attributes in the list $a_v$ for all vertices $v\in
V$, according to Definition~\ref{def2}.
\end{theorem}
\begin{proof}
Let us start with the winding phase of the recursion, that is,
steps~\ref{algo2:deg},~\ref{algo2:minusv} and~\ref{algo2:recur} are
executed until we reach an empty graph $G'$ (see~\ref{algo2:endif}).
At this step, the end of recursion (step~\ref{algo2:endif}) is found because each call
to {\tt find\_cliques$(G')$} function is done with a reduced set of vertices
$|V'|=|V-1|$ (and its induced graph), and $G$ has finite size; i.e., the
recursion is done $n$ times. 

Next, we analyze if steps~\ref{algo2:localL} to~\ref{algo2:localL-end}
maintains the list $a_v$ according to Definition~\ref{def:max}. Firstly,
suppose that a certain set $L$ in a subset of one or more sets in $a_v$
list, that is $L\subset S_j\in a_v$. In this case, for every $S_j\in a_v$, the
conditional step~\ref{algo2:local-s} is false because the size of $L$ is
necessarily smaller than $S_j$ because $L\subset S_j$; then the `else' clause
is executed. The condition in step~\ref{algo2:BnotinL} is also `false' by
hypothesis $L\subset S_j$, and again the `else' clause is selected, setting the
$old$ variable as $TRUE$ (notice this variable was initialized as $FALSE$ in
step~\ref{algo2:iniF}). Therefore, the conditional sentence in
step~\ref{algo2:selectL} will be `false' because the value of $old$ is not
changed until step~\ref{algo2:iniF} is executed again, and $L$ will not included
in $a_v$. Secondly, suppose that a certain $L$ in not a subset of any $A_i\in
a_v$, and certain $A_j$ are subsets of $L$. Considering when condition
in~\ref{algo2:local-s} is true, then $new$ will be set `true' only if sets $A_j$
are been considered, and also the set $A_j$ is eliminated form list $a_v$ in
step~\ref{algo2:elimB} (notice that the case $a_v=\emptyset$
conditions~\ref{algo2:local-s} and~\ref{algo2:local-in-L} are true because $B$
size is zero and $B \cap L = B$ when $B=\emptyset$). Finally, for the case in
which $L$ is different to all sets $A_i\in a_v$, $new$ is set to $TRUE$ in
step~\ref{algo2:BmL} or step~\ref{algo2:BnotinL-T} because, either $|A_i|<|L|$ and
$A_i\nsubseteq L$, or $A_i\nsupseteq L$, respectively.   
Consequently, it is shown
that the algorithm maintains list $a_v$ verifying Definition~\ref{def:max}.
 
We continue with the algorithm from the end of the  winding phase.  Let's start
with the unwinding phase when steps~\ref{algo2:valueH} and~\ref{algo2:return} are
executed for the first time: the return of the function will carry a graph with
just the vertex $v$, no edges and $a_v=\emptyset$ (step~\ref{algo2:ini}), that
is $H=(\{v\},\emptyset,a_v)$.  Then, the following instance(s) can add vertices
of degree zero until one instance begins to add the first edges
(edge), getting a star with leaves (one leaf).
At this instance, steps~\ref{algo2:forneig} and~\ref{algo2:neigL} are executed
and step~\ref{algo2:seta} yields $L=\{v,x\}$ because $a_x$ is empty (the vertex
$x$ has no registered neighbors until now). 
Steps~\ref{algo2:localL} to~\ref{algo2:localL-end} are executed yielding as
result $new=TRUE$ and $old=FALSE$ because $a_v=\emptyset$, and getting finally
$a_v=\{\{v,x\}\}$ in step~\ref{algo2:selectL}. Until now, we shown the initial
phase, where the first cliques of size 2 are stored.

%From this point of the algorithm execution any one of the next vertices
%could build a $K_s/s\in\mathbb{N}$ (e.g., $s=3$) because a new vertex
%joining former vertices constituting a $K_{s-1}$ could appear. It is
%worth remarking that it is not possible to build a $K_{s+1}$ at this
%stage (e.g., $s+1=4$). The reason why it is not possible is that there are
%only $K_{s-1}$ and a new vertex just adds edges between this new vertex
%and the present vertices, although this new vertex will never add an
%edge between the present vertices. In this way, the size of new cliques
%is an increasing function (either the maximum clique remains at same
%size or it is increased by one vertex).

We continue to treat the case at any instance in the unwinding phase of the
Algorithm~\ref{cliques}. 

To verify property~(\ref{def2:p1}) of Definition~\ref{def2}, we need to show
that each set $A_i\in a_v$ is a clique. As we shown, the initialization phase
let always sets composed by a vertex $w$ and its neighbors $z\in N(w)$, which
are cliques of size 2. A new joining vertex $v$ neighbor of $w$ will compute
the step~\ref{algo2:seta} with $a_w\neq\emptyset$ as $L\leftarrow \big(N(v)\cap
A) \cup v \cup w$, where can be $A=\{w,z\}$. Thus, if $z$ is also a neighbor of
$v$, the result will be $L=\{v,w,z\}$, which is also a clique. In general, if
the set $A$ is a clique, then the result of step~\ref{algo2:seta} is also a
clique because we find the intersection of $A$ with the neighbors of the actual
connected vertex $v$.  Then, all the sets $A_i\in a_v$ are always cliques; and
as we already shown that they verify the Definition~\ref{def:max}, they are
also maximal.

To conclude the proof is enough to determine if the
property~(\ref{def2:p2}) of Definition~\ref{def2} is obtained by
Algorithm~\ref{cliques}.  The initial phase was already shown some
paragraphs before. steps~\ref{algo2:localL} to~\ref{algo2:localL-end}
assure that a new clique $L$ that the joining vertex $v$ is stored when
it is different to all the previous ones, or any of the previous $A_i\in
a_v$ are a subset of $L$. Then, when the vertex $v$ is reconnected to
the graph, all the maximal cliques are computed and stored, and among
them the maximum clique that this vertex $v$ belongs to at this time of
the algorithm; we call it $K_{max,j}(v)$. Moreover, it is possible that
in a later time of the execution of the algorithm, other vertex $w$
found a greater clique containing the last seen $K_{max,j}(v)$, but this
clique will be always found because it will stored in the last
reconnected vertex belonging to this clique.  Therefore, reading all the
$A_i\in a_v$ for all vertices the maximum clique is obtained.
\end{proof}

\paragraph{\bf Time complexity of Algorithm~\ref{cliques2}}
We firstly analyze the cost of the central loop in steps~\ref{algo2:localL}
to~\ref{algo2:localL-end}. Considering that the size of list $a_v$ can
be bound by the function $S(d_{\max})$, and the set operations in
steps~\ref{algo2:local-in-L} and~\ref{algo2:BnotinL} are bound by the maximum
clique size in $v$, that is $O(d_{max})$, therefore this loop is done in
$O(d_{max}\cdot S(d_{\max}))$.  

Then, step~\ref{algo2:seta} of Algorithm~\ref{cliques} is an intersection of
two sets of size $d_{\max}$ (maximum degree), if both are ordered\footnote{This
can be done at the beginning for all vertices in graph $G$, taking $O(n\cdot
d_{\max} \cdot \log(d_{\max}))$.} taking $O(d_{\max})$.  Next, we consider
loop~\ref{algo2:neigL}, taking $S(d_{\max})$ times because $|a_x|\leq
S(d_{\max})$, giving $O(d_{\max}\cdot S^2(d_{\max}))$. Loop~\ref{algo2:forneig} on vertex $v$
neighbors takes an extra $d_{\max}$ times, giving $O(d_{\max}^2\cdot
S^2(d_{\max}))$.
Step~\ref{algo2:deg} has a complexity of $O(1)$.  Finally, the recursion is
done for each vertex of the graph, producing a total time complexity of
$O(n\cdot d_{\max}^2\cdot S^2(d_{\max}))$.  Considering graphs with $m\gg n$, we can express the
$n$ recursions and the visit to all neighbors in loop~\ref{algo2:forneig} as
visiting all the edges, yielding a time complexity of $O(m\cdot
d_{\max}\cdot S^2(d_{\max}))$.
%If we also considering that $d_{\max}^2=O(m)$, we get $O(m^{3+\frac{1}{2}})$.

For graphs in general, where $d_{\max}$ could be bound by $n$, time
complexity is $O(n^3\cdot S^2(n))$. Now, the main problem is to bound
$S(d_{\max})$ function. In general, this function can be exponential but for
some family of graphs it can be polynomial. This last case is observed in
most of the graphs issues form the Complex System field, whose have a
heavy tailed degree distribution which can be bound by a power law.
These last have also the property that $m\ll n^2$, that is they are
sparse. As the number of cliques is highly related to the number of
triangles in the graph, and this is low because the graph is sparse, the
number of cliques per vertex is also low. For these cases, we can model
this function as $S(d_{\max})=d_{\max}^\alpha$ with $\alpha\in\mathbb{N}$. Then,
knowing that $n^{\frac{1}{\beta}}$ is a bound of $d_{\max}$,
therefore the complexity yields $O(n^{1+\frac{2+\alpha}{\beta}})$ for
$\beta\leq 2+\alpha$; and for $\beta>2+\alpha$ 
the elimination in \ref{algo:minusv} can dominate, reaching the bound
$O(n^2)$ or $O(n\cdot d_{\max}\cdot\log d_{\max})$.

\paragraph{\bf Storage complexity} It can be computed as the space to storage
graph $G$, which is $O(n\cdot d_{\max})$, and the space occupied by all the
$a_v$ list of sets. This last quantity can be computed as the length of each
set in $a_v$, which is bound by $S(d_{\max})$. Thus, counting $d_{max}$
neighbors, the total storage complexity is $O(n\cdot d_{\max}\cdot
\max(d_{\max},S(d_{\max})))$.

\section{Applications}

In this section we illustrate Algorithm~\ref{cliques} and~\ref{cliques2}
through an implementation developed in {\tt python} programming language
~\cite{ah2011}, and its application to some graphs showing their maximum
clique.

Firstly, we apply our algorithm to find cliques to some random
graphs defined by~\cite{ErdosRenyi59}. It is shown
in~\cite{bollobas01,BollobasEdos76} that an ER random graph has a high
probability to contain a clique of size,
\begin{equation}
r = \frac{2\cdot\log n}{\log 1/p} \enspace,\label{eq_r}
\end{equation}
where $n$ is the number of vertices and $p$ is the probability that an
edge exist between any pair of vertices.

\begin{table}[h]
{\small
\begin{tabular}{ c|c|c|c|c|c|c }
$n$ & $p$ & $\bar{d}$ & $r$ & $|K_{\max}|$ & $|\mathcal{K}|: \mathcal{K}=\{K \in K_{\max}\}$ & induced $r+1$  \\
\hline
\hline
100 & 0.01 & 1 & 2 & {\bf 2}  & 60 & yes \\
1000 & 0.01 & 10  & 3 & {\bf 3}  & 159 & yes \\
10000 & 0.01 & 100  & 4 & {\bf 4}  & 372 & yes \\
10000 & 0.04642 & 464.2 & 6 & {\bf 6}  &  5 & yes \\
\end{tabular}
}
\caption{Maximum cliques in Erd\"os Renyi graphs.}\label{ER_res}
\end{table}

Table~\ref{ER_res} show the results, where the columns are: the size of the
graph, the probability $p$, the average degree $\bar{d}$, 
%(it is close to $d_{\max}$ because ER graphs have an asymptotically Poisson
%degree distribution for certain range of $p$), 
the computed $r$ according to Equation~\ref{eq_r}, the size found by
Equation~\ref{eq_MK}, the number of different cliques (i.e., at least
one vertex is different), and if an induced clique of size $r+1$ were
found. The last column is obtained adding new edges to build a greater
clique than the maximum, it shows 'yes' when this clique is found and
'not' if this is not found. Moreover, the 'yes' answer also means that
we find just one clique of that size (see the number of clique $r$ find
in the original graph).  We tested the algorithm on several graphs of
each kind, obtaining the same results (excluding $|\mathcal{K}|$ which
changed some times). We display just one result of each kind. 

The result is evident, we always find the predicted maximum clique,
even when an artificial one is introduced.

\begin{figure}[b]
\includegraphics[height=0.90\textwidth,angle=-0]{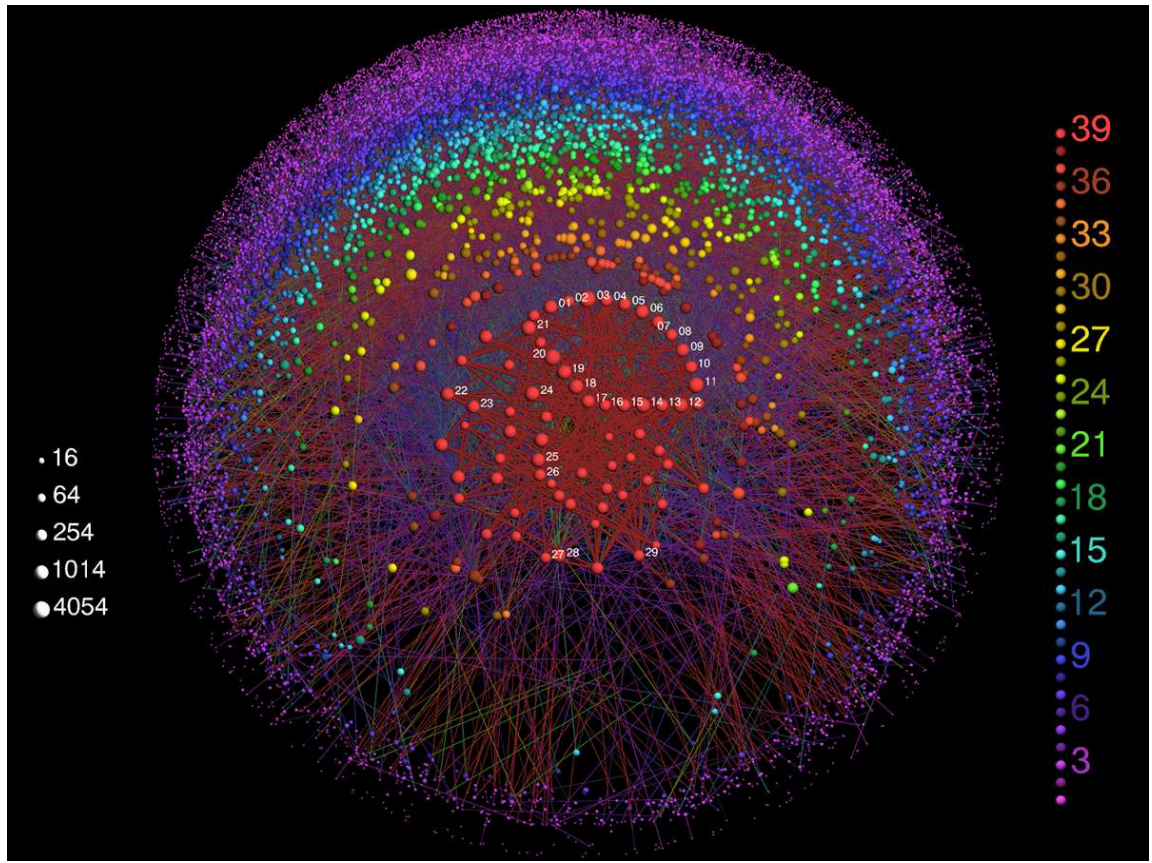}
\caption{Visualization of AS Internet map by LaNet-vi. Vertices in the maximum clique are labeled with numbers.}
\label{fig:AS}
\end{figure}

Secondly, we applied our algorithm to a AS Internet graph. This graph
has, as main properties, a power law degree distribution and most
vertices of low degree are connected to the high degree ones. We used an
exploration of~\cite{IR_CAIDA} performed in September 2011.
Figure~\ref{fig:AS} shows a visualization of this map obtained by
LaNet-vi~\cite{BAHB2008}. This visualization is based on $k$-core
decomposition. A $k$-core is a the maximum induced subgraph such that
all vertices have at least $k$ degree~\cite{Seidman83,Bollobas84}.
LaNet-vi paints each vertex with the rain-bow colors according to its
{\em shell index}, i.e., the maximum core that a vertex belongs to.  It
also makes a greedy clique decomposition of the top core, i.e., the core
with maximum $k$, placing each clique in circular sector according to
its size.  

In this graph Algorithm~\ref{cliques2} found a $K_{29}$ while LaNet-vi
found a $K_{24}$; this is displayed as the largest circular sector of
red vertices in Figure~\ref{fig:AS}. Moreover, this figure shows
vertices of the $K_{29}$ as those enumerated from 01 to 29. It is
possible to appreciate that, even if all vertices are in the top
core, the heuristic of LaNet-vi do not find this clique. For cases
where some vertex is not at the top core LaNet-vi never find the maximum
clique. Algorithm~\ref{cliques} runs several times faster than
Algorithm~\ref{cliques2}, but it not always find the maximum clique. For
instance, for some starting vertices, Algorithm~\ref{cliques} found a
$K_{27}$ instead of $K_{29}$.

Comparing the execution time of this graph and a ER graph of the same
size, e.g., the same number of edges, we find that AS graph ends quicker
than ER graph, since that its degree distribution follows a power law
with $\beta\simeq 2.2$.

\section{Discussion}

As we have already remarked this problem is NP-Complete, and this can be
see from the complexity of Algorithm~\ref{cliques2}, which depends on
the number of cliques that a vertex belongs to, or the
Algorithm~\ref{cliques} do not find always the maximum clique. 
However, this papers aims to introduce a new approach to find cliques,
that seems to be faster that the classical algorithms.

\bibliographystyle{apalike}
\bibliography{clique}

\end{document}